\documentstyle[12pt]{article}

\newcommand{\mn}{\sl}
\def\afterthmseparator{}
\makeatletter
\renewcommand{\@begintheorem}[2]{\trivlist
      \item[\hskip \labelsep{\bf #1\ #2\unskip\afterthmseparator}]\mn}
\renewcommand{\@opargbegintheorem}[3]{\trivlist
      \item[\hskip \labelsep{\bf #1\ #2\ (#3)\unskip\afterthmseparator}]\mn}
\makeatother
\newtheorem{theorem}{Theorem}[section]
\newtheorem{lemma}[theorem]{Lemma}
\newtheorem{corollary}[theorem]{Corollary}
\newtheorem{proposition}[theorem]{Proposition}
\newtheorem{rem}[theorem]{Remark}

\newtheorem{que}[theorem]{Question}
\newenvironment{question}{\renewcommand{\mn}{\rm} \begin{que}}{\end{que}}

\newcommand{\qed}{$\;\;\;\Box$}
\newenvironment{proof}{\par\smallbreak{\sl Proof.~}}
{\unskip\nobreak\hfill \qed \par\medbreak}
\newenvironment{claim}{\par\smallbreak{\it Claim:}}{}

\newcommand{\of}[1]{\left( #1 \right)}
\newcommand{\function}[2]{:#1 \rightarrow #2}

\newcommand{\setdef}[2]{\left\{ \hspace{0.5mm} #1 :
\hspace{0.5mm} #2 \right\}}

\newcommand{\PP}[1]{ {\bf P} \left[ #1 \right] }
\newcommand{\EE}[1]{ {\bf E} \left[ #1 \right] }

\newcommand{\refeq}[1]{(\ref{#1})}

\newcommand{\ssc}{\small\sc}
\newcommand{\simi}{{\ssc Similitude of Permutation Groups}}
\newcommand{\conj}{{\ssc Group Conjugacy}}
\newcommand{\greq}{{\ssc Equality of Permutation Groups}}
\newcommand{\griso}{{\ssc Isomorphism of Permutation Groups}}
\newcommand{\elconj}{{\ssc Element Conjugacy}}
\newcommand{\centr}{{\ssc Centralizer and Coset Intersection}}
\newcommand{\problem}[3]{\par\smallskip\par\noindent #1\\
{\it Given:} #2.\\ {\it Recognize if:} #3. \par\smallskip\par}
\newcommand{\view}{\mbox{view}}
\newcommand{\completeness}{\par{\it Completeness.\ }}
\newcommand{\soundness}{\par{\it Soundness.\ }}
\newcommand{\zk}{\par{\it Zero-knowledge.\ }}
\newcommand{\spn}[1]{\langle #1 \rangle}
\newcommand{\ips}[1]{\{ #1 \}}

\newsavebox{\biblio}
\savebox{\biblio}{\parbox{\textwidth}{\normalsize This paper was published in
{\it Visn.\ L'viv.\ Univ., Ser.\ Mekh.-Mat.}
(Bulletin of the Lviv University, Series in Mechanics and Mathematics)
Vol.61, pp.195--205 (2003).\\[1.5mm]
Reviewed in {\it Zentralblatt f\"ur Mathematik} Zbl 1035.03533.}}

\title{
\mbox{}\\[-45mm]\usebox{\biblio}\\[5mm]
Zero-Knowledge Proofs of the Conjugacy\\
for Permutation Groups}

\author{Oleg Verbitsky\\
{\small Department of Algebra}\\
{\small Faculty of Mechanics \& Mathematics}\\
{\small Kyiv National University}\\
{\small Volodymyrska 60}\\
{\small 01033 Kyiv, Ukraine}
}

\date{}

\begin{document}

\maketitle

\begin{abstract}
We design a perfect zero-knowledge proof system for recognition if
two permutation groups are conjugate. It follows, answering a question
posed by O.~G.~Ganyushkin, that this recognition problem is not
NP-complete unless the polynomial-time hierarchy collapses.
\end{abstract}

\section{Introduction}

Let $S_m$ be a symmetric group of order $m$.
We suppose that an element of $S_m$, a permutation of the set
$\{1,2,\ldots,m\}$, is encoded by a binary string of length
$l=\lceil\log_2m!\rceil$, $m(\log_2m-O(1))\le l\le m\log_2m$.
Given $v\in S_m$, $y\in S_m$, and $Y\subseteq S_m$, we denote
$y^v=v^{-1}yv$ and $Y^v=\setdef{y^v}{y\in Y}$. Two subgroups
$G$ and $H$ of $S_m$ are {\em similar\/} if their actions
on $\{1,2,\ldots,m\}$ are isomorphic or, equivalently, if $G=H^v$ for
some $v\in S_m$. If $X\subseteq S_m$, let $\spn X$ denote the group
generated by elements of $X$.

We address the following algorithmic problem.
\problem{\simi}
{$A_0, A_1\subseteq S_m$}
{$A_0$ and $A_1$ are similar}
Note that the \greq\/ problem, that is, recognition if
$\spn{A_0}=\spn{A_1}$ reduces to recognition, given $X\subseteq S_m$ and
$y\in S_m$, if $y\in \spn X$. Since the latter problem is known
to be solvable in time bounded by a polynomial of the input length
\cite{Sim,FHL}, so is \greq. As a consequence, \simi\/ belongs to
NP, the class of decision problems whose yes-instances have
polynomial-time verifiable certificates. The similitude of
$\spn{A_0}$ and $\spn{A_1}$ is certified by a permutation $v$
such that $\spn{A_1}=\spn{A_0^v}$.

Another problem, \griso, is to recognize if $\spn{A_0}$ and $\spn{A_1}$
are isomorphic. This problem also belongs to NP (E.~Luks, see
\cite[Corollary 4.11]{Bab:dm}). Furthermore, it is announced
\cite{BKL} that \griso\/ belongs to the complexity class coAM
(see Section \ref{s:prel} for the definition). By \cite{BHZ} this
implies that \griso\/ is not NP-complete unless the polynomial-time
hierarchy collapses to its second level (for the background on computational
complexity theory the reader is referred to~\cite{GJo})

O.~G.~Ganyushkin \cite{Gan} posed a question if a similar non-completeness
result can be obtained for \simi. In this paper we answer this question in
affirmative. We actually prove a stronger result of independent interest,
namely, that \simi\/ has a perfect zero-knowledge interactive proof system.
It follows by \cite{AHa} that \simi\/ belongs to coAM and is therefore
not NP-complete unless the polynomial-time hierarchy collapses.

Informally speaking, a zero-knowledge proof system for a recognition
problem of a language $L$ is a protocol
for two parties, the prover and the verifier, that allows the prover to
convince the verifier that a given input belongs to $L$,
with high confidence but without communicating the verifier any information
(the rigorous definitions are in Section~\ref{s:prel}).
Our zero-knowledge proof system for \simi\/ uses the underlying ideas
of the zero-knowledge proof systems designed in \cite{GMR} for the
{\ssc Quadratic Residuosity} and in \cite{GMW} for the {\ssc Graph Isomorphism}
problem. In particular, instead of direct proving something about
the input groups $\spn{A_0}$ and $\spn{A_1}$, the prover prefers to
deal with their conjugates $\spn{A_0}^w$ and $\spn{A_1}^w$ via
a random permutation $w$. The crucial point is that these random groups
are indistinguishable by the verifier because they are identically
distributed, provided $\spn{A_0}$ and $\spn{A_1}$ are similar.
However, we here encounter a complication: the verifier
may actually be able to distinguish between $\spn{A_0}^w$ and
$\spn{A_1}^w$ based on particular representations of these groups
by their generators. Overcoming this complication, which does not
arise in \cite{GMR,GMW}, is a novel ingredient of our proof system.

Our result holds true even for a more general problem of recognizing
if $\spn{A_0}$ and $\spn{A_1}$ are conjugated via an element of
the group generated by a given set $U\subseteq S_m$. We furthermore
observe that a similar perfect zero-knowledge proof system works
also for the \elconj\/ problem of recognizing, given $a_0,a_1\in S_m$
and $U\subseteq S_m$, if $a_1=a_0^v$ for some $v\in\spn U$.
A version of this problem where $a_0,a_1\in\spn U$ was proved to
be in coAM in \cite[Corollary 12.3 (i)]{Bab:dm}. Note that the proof
system developed in \cite{Bab:dm} uses different techniques and is not
zero-knowledge.

\section{Preliminaries}\label{s:prel}

Every decision problem under consideration can be represented through
a suitable encoding as a recognition problem for a language $L$
over the binary alphabet. We denote the {\em length\/} of a binary word
$w$ by $|w|$.

An {\em interactive proof system\/} $\ips{V,P}$, further on
abbreviated as IPS, consists of two probabilistic Turing machines,
a polynomial-time {\em verifier\/} $V$ and
a computationally unlimited {\em prover\/} $P$.
The input tape is common for the verifier and the prover.
The verifier and the prover also share a communication tape
which allows message exchange between them.
The system works as follows. First both the machines $V$ and $P$
are given an input $w$ and each of them is given an individual
random string, $r_V$ for $V$ and $r_P$ for $P$. Then $P$ and $V$
alternatingly write messages to one another in the communication tape.
$V$ computes its $i$-th message $a_i$ to $P$ based on the input $w$,
the random string $r_V$, and all previous messages from $P$ to $V$.
$P$ computes its $i$-th message $b_i$ to $V$ based on the input $w$,
the random string $r_P$, and all previous messages from $V$ to $P$.
After a number of message exchanges $V$ terminates interaction and
computes an output based on $w$, $r_V$, and all $b_i$. The output
is denoted by $\ips{V,P}(w)$.
Note that, for a fixed $w$, $\ips{V,P}(w)$ is a random variable
depending on both random strings $r_V$ and $r_P$.

Let $\epsilon(n)$ be a function of a natural argument taking on
positive real values.
We say that $\ips{V,P}$ is an {\em IPS
for a language $L$ with error $\epsilon(n)$\/} if the following two
conditions are fulfilled.

\smallskip

\noindent
{\em Completeness.} If $w\in L$, then $\ips{V,P}(w)=1$
with probability at least $1-\epsilon(|w|)$.\\
{\em Soundness.} If $w\notin L$, then, for an arbitrary interacting
probabilistic Turing machine $P^*$, $\ips{V,P^*}(w)=1$ with probability
at most $\epsilon(|w|)$.

\smallskip

\noindent
We will call any prover $P^*$ interacting with $P$ on input
$w\notin L$ {\em cheating}.
If in the completeness condition we have $\ips{V,P}(w)=1$
with probability 1, we say that $\ips{V,P}$ has {\em one-sided
error\/} $\epsilon(n)$.

An IPS is {\em public-coin\/} if the concatenation
$a_1\ldots a_k$ of the verifier's messages is a prefix of his random
string $r_V$. A {\em round\/} is sending one message from the verifier
to the prover or from the prover to the verifier.
The class AM consists of those languages having
IPSs with error $1/3$ and with number of rounds bounded by a constant
for all inputs. A language $L$ belongs to the class coAM iff its
complement $\{0,1\}^*\setminus L$ belongs to AM.

\begin{proposition}[Goldwasser-Sipser \cite{GSi}]\label{prop:gsi}
Every IPS for a language $L$ can be converted into a public-coin IPS
for $L$ with the same error at cost of increasing the number of rounds
in 2.
\end{proposition}

Given an IPS $\ips{V,P}$ and an input $w$,
let $\view_{V,P}(w)=(r'_V,a_1,b_1,\ldots,a_k,b_k)$ where $r'_V$ is a part
of $r_V$ scanned by $V$ during work on $w$ and $a_1,b_1,\ldots,a_k,b_k$ are
all messages from $P$ to $V$ and from $V$ to $P$
($a_1$ may be empty if the first message
is sent by $P$). Note that the verifier's messages $a_1,\ldots,a_k$
could be excluded because they are efficiently computable from the other
components. For a fixed $w$, $\view_{V,P}(w)$ is a random variable depending
on $r_V$ and $r_P$.

An IPS $\ips{V,P}$ is {\em perfect
zero-knowledge on $L$\/} if for every interacting polynomial-time
probabilistic Turing machine $V^*$
there is a probabilistic Turing machine $M_{V^*}$, called a {\em simulator},
that on every input $w\in L$ runs in expected polynomial time
and produces output $M_{V^*}(w)$ which,
if considered as a random variable depending on a random string
of $M_{V^*}$, is distributed identically with $\view_{V^*,P}(w)$.
This notion formalizes the claim that the verifier gets no information
during interaction with the prover: everything that the verifier gets
he can get without the prover by running the simulator.
According to the definition, the verifier learns nothing even if he
deviates from the original program and follows an arbitrary probabilistic
polynomial-time program $V^*$.  We will call the verifier $V$ {\em honest\/}
and all other verifiers $V^*$ {\em cheating}.
If, for all $V^*$, $M_{V^*}$ is implemented by the same simulator
$M$ running $V^*$ as a subroutine, we say that $\ips{V,P}$ is
{\em black-box simulation\/} zero-knowledge.

We call $\epsilon(n)$ {\em negligible\/} if $\epsilon(n)<n^{-c}$
for every $c$ and all $n$ starting from some $n_0(c)$.
The class of languages $L$ having IPSs that are perfect
zero-knowledge on $L$ and have negligible error is denoted by PZK.

\begin{proposition}[Aiello-H\aa stad \cite{AHa}]\label{prop:aha}
$\mbox{PZK}\subseteq\mbox{coAM}$.
\end{proposition}

The {\em $k(n)$-fold sequential composition\/} of an
IPS $\ips{V,P}$ is the IPS
$\ips{V',P'}$ in which $V'$ and $P'$ on input $w$ execute
the programs of $V$ and $P$ sequentially $k(|w|)$ times, each time
with independent choice of random strings $r_V$ and $r_P$.
At the end of interaction $V'$ outputs 1 iff $\ips{V,P}(w)=1$
in all $k(|w|)$ executions. The initial system $\ips{V,P}$
is called {\em atomic}.

\begin{proposition}\label{prop:seqrep}
\mbox{}
\begin{enumerate}
\item
If $\ips{V',P'}$ is the $k(n)$-fold sequential composition of $\ips{V,P}$,
then
$$
\max_{P^*}\PP{\ips{V',P^*}(w)=1}=
\of{\max_{P^*}\PP{\ips{V,P^*}(w)=1}}^{k(|w|)}.
$$
Consequently, if $\ips{V,P}$ is an IPS for a language $L$ with one-sided
constant error $\epsilon$, then $\ips{V',P'}$ is an IPS for $L$
with one-sided error $\epsilon^{k(n)}$.
\item
(Goldreich-Oren \cite{GOr}, see also \cite[Lemma 6.19]{Gol})
If in addition $\ips{V,P}$ is black-box simulation perfect zero-knowledge
on $L$, then $\ips{V',P'}$ is perfect zero-knowledge on $L$.
\end{enumerate}
\end{proposition}

In the {\em $k(n)$-fold parallel composition\/} $\ips{V'',P''}$
of $\ips{V,P}$, the program of $\ips{V,P}$
is executed $k(|w|)$ times in parallel, that is, in each round
all $k(|w|)$ versions of a message are sent from one machine to another
at once as a long single message. In every parallel execution
$V''$ and $P''$ use independent copies of $r_V$ and $r_P$.
At the end of interaction $V''$ outputs 1 iff $\ips{V,P}(w)=1$
in all $k(|w|)$ executions.

\begin{proposition}\label{prop:parrep}
If $\ips{V'',P''}$ is the $k(n)$-fold parallel composition of $\ips{V,P}$,
then
$$
\max_{P^*}\PP{\ips{V'',P^*}(w)=1}=
\of{\max_{P^*}\PP{\ips{V,P^*}(w)=1}}^{k(|w|)}.
$$
\end{proposition}

\section{Group Conjugacy}

We consider the following extension of \simi.
\problem{\conj}
{$A_0,A_1,U\subseteq S_m$}
{$\spn{A_1}=\spn{A_0}^v$ for some $v\in\spn U$}

\begin{theorem}\label{thm:grconj}
\conj\/ is in PZK.
\end{theorem}

Designing a perfect zero-knowledge interactive proof system for \conj,
we will make use of the following facts due to Sims~\cite{Sim,FHL}.
\begin{enumerate}
\item
There is a polynomial-time algorithm that, given $X\subseteq S_m$
and $y\in S_m$, recognizes if $y\in\spn X$. As a consequence,
there is a polynomial-time algorithm that, given $X\subseteq S_m$
and $Y\subseteq S_m$, recognizes if $\spn X=\spn Y$.
\item
There is a probabilistic polynomial-time algorithm that, given
$X\subseteq S_m$, outputs a random element of $\spn X$.
Here and further on, by a {\em random element\/} of a finite set $Z$ we mean
a random variable uniformly distributed over $Z$.
\end{enumerate}

Given $A\subseteq S_m$ and a number $k$, define
$$
G(A,k)=\setdef{(x_1,\ldots,x_k)}{x_i\in S_m, \spn{x_1,\ldots,x_k}=\spn A}.
$$

In the sequel, the length of the binary encoding of an input
$A_0,A_1,U\subseteq S_m$ will be denoted by $n$. We set $k=4m$.
On input $(A_0,A_1,U)$, the IPS we design is the
$n$-fold sequential repetition of the following 3-round system.
We will say that the verifier $V$ {\em accepts\/} if $\ips{V,P}(A_0,A_1,U)=1$
and {\em rejects\/} otherwise.

\smallskip

If $(A_0,A_1,U)$ is yes-instance of \conj, $P$ finds an element
$v\in\spn U$ such that $\spn{A_1}=\spn{A_0}^v$.

\smallskip

{\it 1st round.}

\noindent
$P$ generates a random element $u\in\spn U$, computes $A=A_1^u$,
chooses a random element $(a_1,\ldots,a_k)$ in $G(A,k)$,
and sends $(a_1,\ldots,a_k)$ to $V$. $V$ checks if all $a_i\in S_m$
and, if not (this is possible in the case of a cheating prover),
halts and rejects.

\medskip

{\it 2nd round.}

\noindent
$V$ chooses a random bit $\beta\in\{0,1\}$ and sends it to $P$.

\smallskip

{\it 3rd round.}

\noindent
{\it Case $\beta=1$.}
$P$ sends $V$ the permutation $w=u$. $V$ checks if
$w\in\spn U$ and if $\spn{a_1,\ldots,a_k}=\spn{A_1^w}$.

\noindent
{\it Case $\beta\ne 1$} (this includes the possibility of a message
$\beta\notin\{0,1\}$ produced by a cheating verifier).
$P$ computes $w=vu$ and sends $w$ to $V$. $V$ checks if
$w\in\spn U$ and if $\spn{a_1,\ldots,a_k}=\spn{A_0^w}$.

$V$ halts and accepts if the conditions are checked successfully
and rejects otherwise.

\medskip

We now need to prove that this system is indeed
an IPS for \conj\/ and, moreover, that it is perfect zero-knowledge.

\completeness
To show that the prover is able to follow the above protocol,
we have to check that $G(A,k)\ne\emptyset$ for $k=4m$.
The latter is true by the fact that every subgroup of $S_m$
can be generated by at most $m-1$ elements \cite{Jer}.
If $\spn{A_0}$ and $\spn{A_1}$ are conjugate via an element of $\spn U$
and the prover and the verifier follow the protocol, then
$\spn{a_1,\ldots,a_k}=\spn A=\spn{A_1^u}=\spn{A_0^{vu}}$.
Therefore, the verifier accepts with probability 1 both in the
atomic and the composed systems.

\soundness
Assume that $\spn{A_0}$ and $\spn{A_1}$ are not conjugate
via an element of $\spn U$ and consider an arbitrary cheating prover $P^*$.
Observe that if both $\spn{a_1,\ldots,a_k}=\spn{A_1^u}$ and
$\spn{a_1,\ldots,a_k}=\spn{A_0^w}$ with $u,w\in\spn U$, then
$\spn{A_1}=\spn{A_0}^{wu^{-1}}$.
It follows that $V$ rejects for at least one value of $\beta$ and,
therefore, in the atomic system $V$ accepts with probability at most 1/2.
By Proposition \ref{prop:seqrep}~(1), in the composed system $V$
accepts with probability at most $2^{-n}$.

\zk
We will need the following fact.

\begin{lemma}\label{lem:gen}
Let $G$ be a subgroup of $S_m$ and $a_1,\ldots,a_k$ be random independent
elements of $G$.
\begin{enumerate}
\item
If $k=4m$, then $\spn{a_1,\ldots,a_k}=G$ with probability more than 1/2.
\item
If $k=8m$, then $\spn{a_1,\ldots,a_k}=G$ with probability more than
$1-2^{-m}$.
\end{enumerate}
\end{lemma}

\begin{proof}
We will estimate from above the probability that
$\spn{a_1,\ldots,a_k}\ne G$. This inequality is equivalent with
the condition that all $\spn{a_1}$, $\spn{a_1,a_2}$, \ldots,
$\spn{a_1,\ldots,a_k}$ are proper subgroups of $G$. Assume that this
condition is true. Since every subgroup chain in $S_m$ has length
less than $2m$ \cite{Bab:chain,CST}, less than $2m-1$ inclusions
among $\spn{a_1}\subseteq\spn{a_1,a_2}\subseteq\cdots\subseteq
\spn{a_1,\ldots,a_k}$ are proper. In other words, less than $2m-1$
of the events $a_2\notin\spn{a_1}$, $a_3\notin\spn{a_1,a_2}$, \ldots,
$a_k\notin\spn{a_1,\ldots,a_{k-1}}$ occur. Equivalently, there occur
more than $k-2m$ of the events $a_2\in\spn{a_1}$, $a_3\in\spn{a_1,a_2}$,
\ldots, $a_k\in\spn{a_1,\ldots,a_{k-1}}$.

Let $p=|H|/|G|$ be the maximum density of a proper subgroup $H$ of $G$.
Given $a_1,\ldots,a_i\in G$, define $E(a_1,\ldots,a_i)$ to be an
arbitrary subset of $G$ fixed so that
\begin{description}
\item[(i)]
$E(a_1,\ldots,a_i)$ has density $p$ in $G$, and
\item[(ii)]
$E(a_1,\ldots,a_i)$ contains $\spn{a_1,\ldots,a_i}$ if the latter is
a proper subgroup of $G$.
\end{description}
If $\spn{a_1,\ldots,a_k}\ne G$, there must occur more than
$k-2m$ of the events
\begin{equation}\label{eq:events}
a_2\in E(a_1),\ a_3\in E(a_1,a_2),\ldots,a_k\in E(a_1,\ldots,a_{k-1}).
\end{equation}
It suffices to show that the probability of so many occurrences in
\refeq{eq:events} is small enough. Set $X_i(a_1,\ldots,a_k)$ to be
equal to 1 if $a_{i+1}\in E(a_1,\ldots,a_i)$ and to 0 otherwise.
In these terms, we have to estimate the probability that
\begin{equation}\label{eq:bernul}
\sum_{i=1}^{k-1} X_i > k-2m.
\end{equation}
It is easy to calculate that an arbitrary set of $l$ events in
\refeq{eq:events} occurs with probability $p^l$. Hence the events
\refeq{eq:events} as well as the random variables $X_1,\ldots,X_{k-1}$
are mutually independent, and $X_1,\ldots,X_{k-1}$ are successive
Bernoulli trails with success probability $p$.

If $k=4m$, the inequality \refeq{eq:bernul} implies that strictly more
than a half of all the trails are successful. Since $p\le 1/2$, this
happens with probability less than 1/2 and the item 1 of the lemma follows.

If $k=8m$, the inequality \refeq{eq:bernul} implies
$$
\frac1{k-1}\sum_{i=1}^{k-1}X_i>p+\epsilon
$$
with deviation $\epsilon=1/4$ from the mean value
$p=\EE{\frac1{k-1}\sum_{i=1}^{k-1}X_i}$.
By the Chernoff bound \cite[Theorem A.4]{ASp}, this happens with
probability less than
$\exp\of{-2\epsilon^2(k-1)}\allowbreak =\exp(-m+\frac18)<2^{-m}$.
This proves the item 2 of the lemma.
\end{proof}

By Proposition \ref{prop:seqrep}~(2) it suffices to show that the atomic
system is black-box simulation perfect zero-knowledge.
We describe a probabilistic simulator $M$ that uses
the program of $V^*$ as a subroutine and, for each $V^*$, runs
in expected polynomial time. Assume that the running time
of $V^*$ is bounded by a polynomial $q$ in the input size.
On input $(A_0,A_1,U)$ of length $n$, $M$ will run the program of
$V^*$ on the same input with random string $r$,
where $r$ is the prefix of $M$'s random
string of length $q(n)$. In all other cases of randomization, $M$ will use
the remaining part of its random string.

Having received an input $(A_0,A_1,U)$, the simulator
$M$ chooses a random element $w\in\spn U$ and a
random bit $\alpha\in\{0,1\}$. Then $M$ randomly and independently
chooses elements $a_1,\ldots,a_k$ in $\spn{A_\alpha^w}$ and checks if
\begin{equation}\label{eq:succ}
\spn{a_1,\ldots,a_k}=\spn{A_\alpha^w}.
\end{equation}
If \refeq{eq:succ} is not true, $M$ repeats the choice of $a_1,\ldots,a_k$
again and again until \refeq{eq:succ} is fulfilled.
By Lemma \ref{lem:gen}~(1), $M$ succeeds in at most 2 attempts on average.
The resulting sequence $(a_1,\ldots,a_k)$ is uniformly distributed
on $G(A_\alpha^w,k)$.
Then $M$ computes $\beta=V^*(A_0,A_1,U,r,a_1,\ldots,a_k)$,
the message that $V^*$ sends $P$ in the 2-nd round
after receiving $P$'s message $a_1,\ldots,a_k$.
If $\beta$ and $\alpha$ are simultaneously equal to or different from 1,
$M$ halts and outputs $(r',a_1,\ldots,a_k,\beta,w)$, where $r'$ is the
prefix of $r$ that $V^*$ actually uses after reading the input
$(A_0,A_1,U)$ and the prover's message $a_1,\ldots,a_k$.
If exactly one of $\beta$ and $\alpha$ is equal to 1, then $M$
restarts the same program from the very beginning with another
independent choice of $w$, $\alpha$, and $a_1,\ldots,a_k$.
Notice that it might happen that in unsuccessful attempts
$V^*$ used a prefix of $r$ longer than $r'$.

\medskip

We first check that, for each $V^*$, the simulator $M$ terminates
in expected polynomial
time whenever $A_0$ and $A_1$ are conjugated via an element of $\spn U$.
Since $V^*$ is polynomial-time, one attempt to
pass the body of $M$'s program takes time bounded by a polynomial of $n$.
Observe that $\alpha$ and $(r,a_1,\ldots,a_k)$ are
independent. Really, independently of whether $\alpha=0$ or $\alpha=1$,
$r$ is a random string of length $q(n)$ and $(a_1,\ldots,a_k)$ is
a random element of $G(A,k)$, where $A$ itself is a random element
of the orbit $\setdef{A_0^w}{w\in\spn U}=\setdef{A_1^w}{w\in\spn U}$
under the conjugating action of $\spn U$ on subsets of $S_m$.
It follows that $\alpha$ and $\beta$ are independent and therefore
an execution of the body of $M$'s program is successful with probability 1/2.
We conclude that on average $M$'s program is executed twice
and this takes expected polynomial time.

We finally need to check that, whenever $A_0$ and $A_1$ are conjugated
via an element of $\spn U$, for each $V^*$ the output $M(A_0,A_1,U)$ is
distributed identically with $\view_{V^*,P}(A_0,A_1,U)$.
Notice that both the random variables depend on $V^*$'s random string $r$.
It therefore suffices to show that the distributions are identical
when conditioned on an arbitrary fixed $r$.
Denote these conditional distributions by
$D_M(A_0,A_1,U,r)$ and $D_{V^*,P}(A_0,A_1,U,r)$.
We will show that both $D_M(A_0,A_1,U,r)$ and $D_{V^*,P}(A_0,A_1,U,r)$
are uniform on the set
\begin{eqnarray*}
S=\Bigl\{ \hspace{0.5mm} (a_1,\ldots,a_k,\beta,w)\ : \
w\in\spn U,&&\hspace{-0.8em} \beta=V^*(A_0,A_1,U,r,a_1,\ldots,a_k),\\
&& (a_1,\ldots,a_k)\in G(A^w_{\delta(\beta)},k)\Bigr\},
\end{eqnarray*}
where $\delta(\beta)$ is equal to 1 if $\beta=1$ and to 0 otherwise.

Let $v\in\spn U$, such that $\spn{A_1}=\spn{A_0}^v$, be chosen by
the prover $P$ on input $(A_0,A_1,U)$.
Given $x_1,\ldots,x_k\in G(A_1,k)$ and $u\in\spn U$, define
$\phi(x_1,\ldots,x_k,u)=(a_1,\ldots,a_k,\beta,w)$ by
$a_i=x_i^u$ for all $i\le k$, $\beta=V^*(A_0,A_1,U,r,a_1,\ldots,a_k)$,
and $w=v^{1-\delta(\beta)}u$. As easily seen, $\phi(x_1,\ldots,x_k,u)\in S$.

\begin{claim}{}
The map $\phi\function{G(A_1,k)\times\spn U}S$ is one-to-one.
\end{claim}

\begin{proof}
Define $\psi(a_1,\ldots,a_k,\beta,w)=(x_1,\ldots,x_k,u)$ by
$u=v^{\delta(\beta)-1}w$ and $x_i=a_i^{u^{-1}}$ for all $i\le k$.
It is not hard to check that the map $\psi$ is the inverse of $\phi$.
\end{proof}

Observe now that if $(x_1,\ldots,x_k,u)$ is chosen at random uniformly
in $G(A_1,k)\times\spn U$, then $\phi(x_1,\ldots,x_k,u)$ has
distribution $D_{V^*,P}(A_0,A_1,U,r)$. By Claim we conclude that
$D_{V^*,P}(A_0,A_1,U,r)$ is uniform on $S$.

As a yet another consequence of Claim, observe that if a random tuple
$(a_1,\ldots,a_k,\allowbreak \beta,w)$ is uniformly distributed on $S$,
then its prefix $(a_1,\ldots,a_k)$ is a random element of $G(A,k)$,
where $A$ is a random element of the orbit
$\setdef{A_0^w}{w\in\spn U}=\setdef{A_1^w}{w\in\spn U}$
under the conjugating action of $\spn U$ on subsets of $S_m$.
This suggests the following way of generating a random element of $S$.
Choose uniformly at random $\alpha\in\{0,1\}$, $w\in\spn U$,
$(a_1,\ldots,a_k)\in G(A_\alpha^w,k)$ and, if
\begin{equation}\label{eq:guess}
\delta\of{V^*(A_0,A_1,U,r,a_1,\ldots,a_k)}=\alpha,
\end{equation}
output $(a_1,\ldots,a_k,V^*(A_0,A_1,U,r,a_1,\ldots,a_k),w)$;
otherwise repeat the same procedure once again independently.
Under the condition that \refeq{eq:guess} is fulfilled for the first
time in the $i$-th repetition, the output is uniformly distributed
on $S$. Notice now that this sampling procedure coincides with the
description of $D_M(A_0,A_1,U,r)$. It follows that $D_M(A_0,A_1,U,r)$
is uniform on $S$. The proof of the perfect zero-knowledge property
of our proof system for \conj\/ is complete.

The following corollary immediately follows from Theorem \ref{thm:grconj}
by Proposition \ref{prop:aha} and the result of~\cite{BHZ}.

\begin{corollary}\label{cor:grconj}
\conj\/ is in coAM and is therefore not NP-complete unless the
polynomial-time hierarchy collapses.
\end{corollary}

We also give an alternative proof of this corollary that consists in direct
designing a two-round IPS $\ips{V,P}$ with error 1/4 for the complement of
\conj\/ and applying Proposition \ref{prop:gsi}. More precisely, we
deal with the {\ssc Group Non-Conjugacy} problem of recognizing,
given $A_0,A_1,U\subseteq S_m$, if there is no $v\in\spn U$ such that
$\spn{A_1}=\spn{A_0}^v$.

Set $k=8m$. The below IPS is composed twice in parallel.

\medskip

{\it 1st round.}

\noindent
$V$ chooses a random bit $\alpha\in\{0,1\}$, a random element
$u\in\spn U$, and a sequence of random independent elements
$a_1,\ldots,a_k\in\spn{A_\alpha^u}$.
Then $V$ sends $(a_1,\ldots,a_k)$ to $P$.

\smallskip

{\it 2nd round.}

\noindent
$P$ determines $\beta$ such that $\spn{a_1,\ldots,a_k}$ and $\spn{A_\beta}$
are conjugate via an element of $\spn U$ and sends $\beta$ to $V$.

\smallskip

$V$ accepts if $\beta=\alpha$ and rejects otherwise.

\medskip

\completeness
By Lemma \ref{lem:gen}~(2), $\spn{a_1,\ldots,a_k}=\spn{A_\alpha^u}$
with probability at least $1-2^{-m}$. If this happens and if
$\spn{A_0}$ and $\spn{A_1}$ are not conjugated via $\spn U$, the group
$\spn{a_1,\ldots,a_k}$ is conjugated via $\spn U$ with precisely one of
$\spn{A_0}$ and $\spn{A_1}$. In this case $P$ is able to determine
$\alpha$ correctly. Therefore $V$ accepts with probability at least
$1-2^{-m}$ in the atomic system and with probability at least $1-2^{-m+1}$
in the composed system.

\soundness
If $\spn{A_0}$ and $\spn{A_1}$ are conjugated via $\spn U$, then
for both values $\alpha=0$ and $\alpha=1$, the vector $(a_1,\ldots,a_k)$ has
the same distribution, namely, it is a random element of $A^k$,
where $A$ is a random element
of the orbit $\setdef{A_0^w}{w\in\spn U}=\setdef{A_1^w}{w\in\spn U}$
under the conjugating action of $\spn U$ on subsets of $S_m$.
It follows that, irrespective of his program,
$P$ guesses the true value of $\alpha$ with probability 1/2. With the same
probability $V$ accepts in the atomic system.
By Proposition \ref{prop:parrep}, in the composed system $V$ accepts
with probability 1/4.

Note that $\ips{V,P}$ is perfect zero-knowledge only for the honest verifier
but may reveal a non-trivial information for a cheating verifier.

\section{Element Conjugacy}

This section is devoted to the following problem.
\problem{\elconj}{$a_0,a_1\in S_m$, $U\subseteq S_m$}%
{$a_1=a_0^v$ for some $v\in\spn U$}
L.~Babai \cite{Bab:dm} considers a version of this problem with
$a_0,a_1\in\spn U$ and proves that it belongs to coAM.
His result holds true not only for permutation groups but also for
arbitrary finite groups with efficiently performable group operations,
in particular, for matrix groups over finite fields. It is easy to see that
Theorem \ref{thm:grconj} carries over to \elconj.

\begin{theorem}\label{thm:elconj}
\elconj\/ is in PZK.
\end{theorem}

The proof system designed in the preceding section for \conj\/ applies
to \elconj\/ as well. Moreover, the proof system for \elconj\/ is
considerably simpler. In place of
groups $\spn{A_0^u}$ and $\spn{A_1^u}$ we now deal with single
elements $a_0^u$ and $a_1^u$ and there is no complication with
representation of $\spn{A_0^u}$ and $\spn{A_1^u}$ by generating sets.

We now notice relations of \elconj\/ with the following problem
considered by E.~Luks \cite{Luk} (see also \cite[Section 6.5]{Bab:handbook}).
Given $x\in S_m$, let $C(x)$ denote the centralizer of $x$ in $S_m$.
\problem{\centr}%
{$x,y\in S_m$, $U\subseteq S_m$}{$C(x)\cap\spn U y\ne\emptyset$}

\noindent
Since, given a permutation $x$, one can efficiently find a list of
generators for $C(x)$, this is a particular case of the
{\ssc Coset Intersection} problem of recognizing, given $A,B\subseteq S_m$
and $s,t\in S_m$, if the cosets $\spn As$ and $\spn Bt$ intersect.

\begin{proposition}\label{prop:red}
\elconj\/ and \centr\/ are equivalent with respect to the polynomial-time
many-one reducibility.
\end{proposition}

\begin{proof}
We first reduce \elconj\/ to \centr.
Given permutations $a_0$ and $a_1$, it is easy to recognize if they
are conjugate in $S_m$ and, if so, to find an $s$ such that $a_1=a_0^s$.
The set of all $z\in S_m$ such that $a_1=a_0^z$ is the coset $C(a_0)s$.
It follows that $\spn U$ contains $v$ such that $a_1=a_0^v$ iff
$C(a_0)$ and $\spn Us^{-1}$ intersect.

A reduction from \centr\/ to \elconj\/ is based on the fact that
$C(x)$ and $\spn Uy$ intersect iff $x$ and $yxy^{-1}$ are conjugated
via an element of $\spn U$.
\end{proof}

Note that, while the reduction we described from \elconj\/ to \centr\/
works only for permutation groups, the reduction in the other direction
works equally well for arbitrary finite groups with efficiently performable
group operations, in particular, for matrix groups over finite fields.

We now have three different ways to prove that \elconj\/ is in coAM
and is therefore not NP-complete unless the polynomial-time hierarchy
collapses. First, this fact follows from Theorem \ref{thm:elconj}
by Proposition \ref{prop:aha}. Second, one can use Proposition \ref{prop:red}
and the result of \cite[Corollary 12.2~(d)]{Bab:dm} that {\ssc Coset
Intersection} is in coAM. Finally, one can design a constant-round
IPS for the complement of \elconj\/ as it is done in the preceding
section for the complement of \conj.

We conclude with two questions.

\begin{question}
Is there any reduction between \conj\/ and {\ssc Coset Intersection}?
We are not able to prove an analog of Proposition \ref{prop:red}
for groups because, given $A_0,A_1\subseteq S_m$ such that
$\spn{A_1}=\spn{A_0}^v$ for some $v\in S_m$, we cannot efficiently
find any $v$ with this property (otherwise we could efficiently recognize
the \simi).
\end{question}

\begin{question}
Does \elconj\/ reduce to \conj?
Whe\-re\-as Corollary \ref{cor:grconj} gives us an evidence that
\conj\/ is not NP-complete, we have no formal evidence supporting
our feeling that \conj\/ is not solvable efficiently. A reduction
from \elconj\/ could be considered such an evidence as \elconj\/
is not expected to be solvable in polynomial
time~\cite[page 1483]{Bab:kyoto}.

Note that the conjugacy of permutations $a_0$ and $a_1$ via an element
of a group $\spn U$ does not reduce to the conjugacy of the cyclic
groups $\spn{a_0}$ and $\spn{a_1}$ via $\spn U$ because
$\spn{a_0}$ and $\spn{a_1}$ can be conjugated by conjugation
of another pair of their generators, while such a new conjugation
may be not necessary via $\spn U$. For example, despite the groups
$\spn{(123)}$ and $\spn{(456)}$ are conjugated via
$\spn{(14)(26)(35)}$, the permutations $(123)$ and $(456)$ are not.
\end{question}

\subsubsection*{Acknowledgement}
I appreciate useful discussions with O.~G.~Ganyushkin.

\bigskip\bigskip

\mbox{}\hfill
\parbox{40mm}{Received 15.12.2001\\Accepted 14.03.2003}

\end{document}